\documentclass[11pt]{article}

\usepackage{fullpage}
\usepackage{amsfonts,amssymb,amsmath,amsthm,amsopn,mathrsfs,mathtools,wasysym} 
\usepackage{booktabs,diagbox,colortbl,multirow,tabularx,threeparttable,hhline} 
\usepackage{graphicx} 
\usepackage[caption=false,font=normalsize,labelfont=sf,textfont=sf]{subfig}
\usepackage[listings,skins,breakable]{tcolorbox}

\usepackage{color,xcolor} 
\usepackage{enumerate}
\usepackage{authblk}
\usepackage{footnote}
\usepackage{hyperref}
\usepackage{prettyref}
\usepackage{setspace}
\usepackage{scrpage2}
\usepackage{geometry}
\usepackage{cite}

\usepackage{algorithm, algorithmic}
\usepackage{graphicx}
\usepackage{epstopdf}
\usepackage{float}
\usepackage{subfig}
\definecolor{myblue}{RGB}{34,31,217}
\definecolor{mycyan}{gray}{.7}
\definecolor{Gray}{gray}{0.9}

\newtheorem{theorem}{Theorem}

\newtheorem{lemma}{Lemma}

\newrefformat{fig}{Fig.~\ref{#1}}
\newrefformat{tab}{Table~\ref{#1}}
\newrefformat{sec}{Section~\ref{#1}}
\newrefformat{alg}{Algorithm~\ref{#1}}
\newrefformat{property}{Property~\ref{#1}}
\newrefformat{theorem}{Theorem~\ref{#1}}
\newrefformat{definition}{Definition~\ref{#1}}
\newrefformat{corollary}{Corollary~\ref{#1}}
\newrefformat{lemma}{Lemma~\ref{#1}}
\newrefformat{conj}{Conjecture~\ref{#1}}
\newrefformat{def}{Definition~\ref{#1}}
\newrefformat{eq}{equation~(\ref{#1})}
\newrefformat{app}{Appendix~\ref{#1}}

\usepackage{lscape}

\begin{document}

\title{\vspace{-1ex}\LARGE\textbf{Multidimensional Resource Fragmentation-Aware Virtual Network Embedding in MEC Systems Interconnected by Metro Optical Networks}
}

\author[1]{\normalsize Yingying Guan}
\author[2]{\normalsize Qingyang Song}
\author[3]{\normalsize Weijing Qi}
\author[4]{\normalsize Ke Li}
\author[5]{\normalsize Lei Guo}
\author[6]{\normalsize Abbas Jamalipour}

\affil[1]{\normalsize School of Communication and Information Engineering, Chongqing University of Posts and Telecommunications, Chongqing 400065, P. R. China}
\affil[2,3,5]{\normalsize School of Communication and Information Engineering, Chongqing University of Posts and Telecommunications, Chongqing 400065, P. R. China}
\affil[4]{\normalsize Department of Computer Science, University of Exeter, EX4 4QF, Exeter, UK}
\affil[6]{\normalsize School of Electrical and Information Engineering, University of Sydney, Sydney, NSW 2006, Australia}
\affil[$\ast$]{\normalsize Email: \texttt{songqy@cqupt.edu.cn}}

\date{}
\maketitle

\vspace{-3ex}
{\normalsize\textbf{Abstract: } 
The increasing demand for diverse emerging applications has resulted in the interconnection of multi-access edge computing (MEC) systems via metro optical networks. To cater to these diverse applications, network slicing has become a popular tool for creating specialized virtual networks. However, resource fragmentation caused by uneven utilization of multidimensional resources can lead to reduced utilization of limited edge resources.
To tackle this issue, this paper focuses on addressing the multidimensional resource fragmentation problem in virtual network embedding (VNE) in MEC systems with the aim of maximizing the profit of an infrastructure provider (InP). The VNE problem in MEC systems is transformed into a bilevel optimization problem, taking into account the interdependence between virtual node embedding (VNoE) and virtual link embedding (VLiE). To solve this problem, we propose a nested bilevel optimization approach named BiVNE. The VNoE is solved using the ant colony system (ACS) in the upper level, while the VLiE is solved using a combination of a shortest path algorithm and an exact-fit spectrum slot allocation method in the lower level.
Evaluation results show that the BiVNE algorithm can effectively enhance the profit of the InP by increasing the acceptance ratio and avoiding resource fragmentation simultaneously.
}



\section{Introduction}
\label{sec:introduction}
Multi-access edge computing (MEC) systems, comprising MEC servers and base stations, are increasingly being employed as edge clouds to support various emerging applications that are delay-sensitive and computation-intensive \cite{MECETSI,MaoYZHL17, MachB17}.
A huge amount of data generated by these emerging applications require the collaboration of multiple MEC servers, prompting a lot of frequent information interactions between MEC systems \cite{LimLHJLYNM20}.
These MEC systems are inevitably interconnected by metro optical networks due to their high flexibility and high capacity \cite{ChatterjeeSO15}, which boosts the heterogeneity of network resources and therefore the difficulty and the complexity of allocating these multidimensional resources \cite{HuangYYZC20}.
Network slicing appears to resolve the network resource allocation problem efficiently by providing customized network slices (i.e., virtual networks) for service providers (SPs) to meet the need of diverse applications \cite{YangXGQCC21, WangWYH22}.
When receiving a request for a network slice from an SP, an infrastructure provider (InP) conducts a virtual network embedding (VNE) process, that is, allocates some physical nodes (PNs) and physical links (PLs) to the network slice to satisfy the corresponding resource requirements.
It has been confirmed that the VNE is an NP-hard problem \cite{FischerBBMH13}.

Previous research has extensively investigated the VNE problem in data center networks, considering congestion control \cite{PhamHC20}, energy consumption \cite{EramoMA16,ZhangCSL22,BillingsleyLMMG19,BillingsleyLMMG20,BillingsleyMLMG20,BillingsleyLMMG21}, failure avoidance  \cite{ShahriarACKBMZ20,AyoubBMT22,DehuryS22}, and profit  growth \cite{NguyenH22,SongCGYKZ21,GongJWZ16, ZhaoSB13}.
However, the VNE problem becomes more complex in edge cloud networks, which consist of multiple MEC-equipped base stations (i.e., edge nodes) connected by optical links.
First, limited resources in the edge network.
Second, optical spectrums in the optical links.
Optical spectrum allocation must adhere to spectrum continuity and consistency constraints, resulting in the generation of optical spectrum fragments at links \cite{ChatterjeeWO21}.
Third, different types of resources (i.e., communication and computing resources) provided by MEC-equipped base stations.
Unbalanced utilization of different types of resources leads to resource fragmentation at edge nodes.
As a result, the multidimensional resource fragments generated on links and nodes deteriorate the utilization of inherently limited edge resources, which significantly damages the profit of InP~\cite{ChenLY18,ZouJYZZL19,LiZZL09,LiZLZL09,Li19}.

To improve the efficiency of resource utilization, many works on VNE have been done in the optical and edge networks, respectively.
Regarding edge networks, both radio and computing resources are considered in \cite{joint22TNET, Dynamic20TVT, Time-Sensitive21TNSM}.
The authors of \cite{joint22TNET} investigated the problem of joint assignment of radio and computing resources to network slices, as well as the management of the two types of resources within each slice.
The authors of \cite{Dynamic20TVT} improved the operator's revenue through the joint control of the number of radio channels and the CPU clock speed.
The authors of \cite{Time-Sensitive21TNSM} investigated a service function chain placement problem to minimize service interruption while optimizing radio and computing resource utilization.
However, in those studies, the uneven utilization of radio and computing resources is ignored.
Moreover, these studies lack the consideration of optical link characteristics.
Regarding optical networks, some efforts have been made to avoid optical spectrum fragmentation \cite{ZhuZSLCG18, WeiGWYL19, FanXCCY21}. A common work in this studies is that a parameter related to the utilization of the optical spectrum is first defined, and then this parameter is used to guide the VNE process. In \cite{ZhuZSLCG18}, a fragmentation-aware VNE algorithm is proposed based on a virtual-auxiliary-graph approach.
A parameter related to the fragment size on the attached links of every PN is defined, and then the parameter is used to be the weight for performing virtual node embedding (VNoE) on every virtual-auxiliary-graph.
In \cite{WeiGWYL19}, a matching factor is defined according to the degree of contiguity of the free spectrum on the connected links of each PN.
After seeking the proper PNs by greedily using the matching factor, the virtual link embedding (VLiE) is conducted to find an available path in which the difference between the size of the selected spectrum block and that being requested is the smallest.
In \cite{FanXCCY21}, a VNE method that utilizes complete path evaluation and node proximity sensing is presented.
However, the VNoE and the VLiE are still performed separately, thus the ignorance of the coupling between these two potentially increases the probability of resource fragmentation.

In this paper, with the consideration of the multidimensional resource fragmentation and the dependency between the VNoE and the VLiE, we study the VNE problem in MEC systems interconnected by metro optical networks. Specifically, we jointly embed nodes and links with the goal of maximizing the profit of InP.
There are two key problems tackled.
Firstly, how to avoid the resource fragmentation on both PNs and PLs efficiently?
Secondly, how to take into full consideration of the dependency between VNoE and VLiE during the process of resource allocation? The main contributions of this paper are summarized as follows:
\begin{itemize}
\item[1)]An VNE problem for effective utilization of limited edge cloud resources is studied with the best-effort avoidance of multidimensional resource fragmentation. To quantify the multidimensional resource fragmentation, a level of imbalanced resource utilization and a threshold for judging whether a continuous free spectrum block is a fragment are defined to calculate the resource fragments at the PNs and PLs, respectively.
\item[2)]Considering the dependency between VNoE and VLiE, we transform the VNE problem into a bilevel problem, where the upper layer is the problem of selecting PNs for VNR, and the lower layer is the problem of link selection and spectrum resource allocation.
\item[3)]A nested bilevel VNE method (BiVNE) is proposed to solve the bilevel problem. Specifically, the upper layer problem is solved by a method based on ant colony system (ACS) and the lower layer problem is solved by the Dijkstra algorithm and an exact fit spectrum slot assignment method.
\item[4)]Extensive simulation results validate the performance of BiVNE. Compared with some state-of-the-art algorithms, BiVNE can improve the profit of InP through reducing resource consumption while increasing acceptance ratio.
\end{itemize}

The remainder of the paper is structured as follows. The network model and problem formulation are described in detail in Section \ref{sec:preliminaries}. The problem transformation and the proposed approach are described in Section \ref{sec:proposal}. The simulation setup and corresponding results are presented in section \ref{sec:settings} and section \ref{sec:experiments}. Section \ref{sec:conclusions} concludes this paper.


\section{Preliminaries}
\label{sec:preliminaries}
\subsection{Physical Network}

A physical network topology can be represented as an undirected graph ${G^s}=({N^s},{E^s})$, where ${N^s}$ and ${E^s}$ respectively represent the sets of PNs and PLs which are optical.
 $|N^s|$ is the total number of the PNs.
 As shown in Fig. \ref{fig_VNE_model}, each PN ${n^s}\in{N^s}$ is initially equipped with an MEC server with computing capacity $C(n^s)$ and a macro base station with $W(n^s)$ wireless channels.
The geographical location of $n^s$ is $loc(n^s)$.
The entire spectrum in each optical link ${e^s } \in {E^s}$ is divided into $B(e^s)$ granular frequency slots (FSs).
The maximum available adjacent slot block (MACSB) $m_{e^s}\in M_{e^s}$ is the block of more than one unoccupied FS that exists in the optical spectrum.
We define ${P^s}$ as a set of circle-free paths in ${G^s}({V^s},{E^s})$.
The available bandwidth on path ${p^s \in P^s}$ can be computed by ${B(p^s) = {\min} \{B(e^s)\mid I_{e^s}^{p^s} = 1 \}}$, where $I_{e^s}^{p^s}$ is a binary variable indicating whether ${p^s}$ passes optical link ${e^s}$ (${I_{e^s}^{p^s} = 1}$), or not (${I_{e^s}^{p^s} = 0}$).

\begin{figure}[ht]
    \centering
    \includegraphics[width=80mm]{{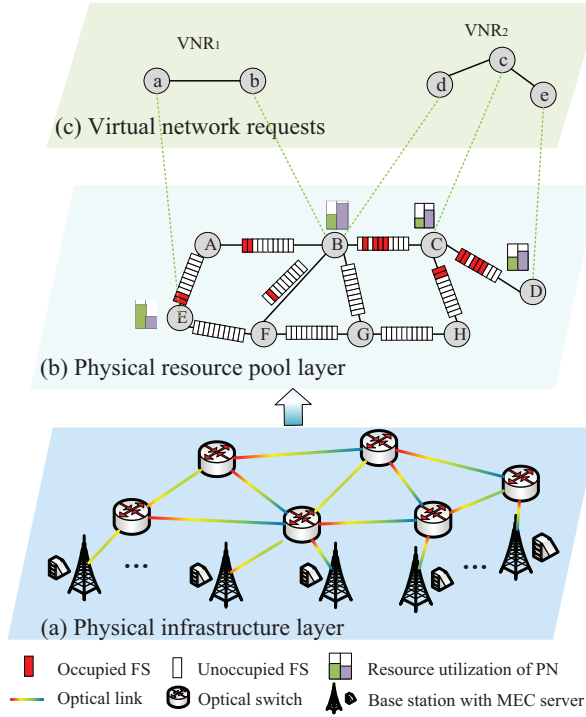}}
    \vspace{0cm}
    \caption {VNE example. }
    \label{fig_VNE_model}
\end{figure}

\subsection{Virtual Network Request (VNR)}
Similar to the physical network topology, a VNR topology can be also represented as an undirected graph ${G^r}=({N^r},{E^r})$. The set of VNRs denoted as $\Upsilon$. ${N^r}$ is the set of virtual nodes (VNs), and ${E^r}$ represents the set of virtual links (VLs). ${n^r}\in{N^r}$ represents a VN. $|N^r|$ is the total number of VNs in ${G^r}$. Each VN ${n^r}$ not only requires the amount of computing resource $C(n^r)$ and the number of wireless channels $W(n^r)$, but also prefers a corresponding physical area, where $loc(n^r)$ and $\Delta loc(n^r)$ are respectively the center and the radius of the area.
Thus, ${n^r}$ only can be embedded onto a set of candidate PNs within the preferred area, i.e.
$\Phi_N(n^r) = \{n^s\in{N^s} \mid dis[loc(n^r),loc(n^s)] \le \Delta loc(n^r)\}$,
 where $dis[loc(n^r),loc(n^s)]$ represents the distance between $n^r$ and $n^s$.
The amount of FSs required by VL ${e^r}$ is $B(e^r)$.
${s(e^r)}$ and ${t(e^r)}$ denote the end nodes of ${e^r}$.
Each VL ${e^r}$ is associated with a candidate physical path set, $\Phi_E(e^r) = \{p^s \mid B(p^s) \geq B(e^r),\  {s(p^s)} \in \Phi_N(s(e^r)),\  {t(p^s)}\in \Phi_N(t(e^r)) \}$. The main notations used in the paper are summarized in Table \ref{tab_notions}.

\begin{table}[ht]    
\caption{List of Notations}  
\centering  
   \begin{tabular}{c|lp{10cm}p{2cm}}   
        \toprule
        \textbf{Notation}       &\textbf{Description}\\
        \midrule

         ${G^s}=({N^s},{E^s})$       &Graph for physical network\\

         $n^s$         &Physical node, ${n^s}\in{N^s}$\\

         $e^s$         &Physical link, ${e^s } \in {E^s}$\\

         $loc(n^s)$    &Geographical location of ${n^s}$ \\

         $C(n^s)$      &Computing capacity of $n^s$ \\
         $C_l(n^s)$    &Occupied computing capacity of $n^s$\\

         $W(n^s)$      &Total number of wireless channels of $n^s$\\
         $W_l(n^s)$    &Number of occupied wireless channels of $n^s$\\

         $B(e^s)$      &Bandwidth of $e^s$ \\
         $P^s$         &Set of circle-free paths in ${G^s}$\\

         $M_{e^s}$     &Set of MACSBs on $e^s$, $m_{e^s}\in M_{e^s}$\\

         $s(m_{e^s})$,$t(m_{e^s})$  &Starting and ending FS indices of $m_{e^s}$ \\

         $f(m_{e^s})$  &Number of FSs in $m_{e^s}$ \\

         $\Xi_{max}$   &Maximum fragment size\\

         $\xi_{e^s}^k$ &$k_{th}$ spectral fragment on $e^s$\\

         ${G^r}=({N^r},{E^r})$       &Graph for VNR\\

         $C(n^r)$      &Computing resource required by $n^r$ \\

         $W(n^r)$      &Number of wireless channels required by $n^r$\\

         $loc(n^r)$    &Center of $n^r$'s preference area\\

         $\Delta loc(n^r)$ &Radius of $n^r$'s preference area \\

         $ B(e^r) $        &Amount of bandwidth required by $e^r$\\

         ${s(e^r)}$, ${t(e^r)}$    &End nodes of $e^r$ \\

         $I_{e^s}^{p^s}$    &Binary variable,  equals 1 if ${p^s}$ passes ${e^s}$\\

         $\nu^r$           &Binary variable, equals 1 if ${G^r}$ is successfully \\ & embedded\\

         ${x_{n^s}^{n^r}}$  &Binary variable, equals 1 if ${n^r}$ is embedded onto ${n^s}$\\

         ${y_{p^s}^{e^r}}$  &Binary variable, equals 1 if ${e^r}$ is embedded onto $p^s$\\

         ${z_{e^s,b}^{e^r}}$        &Binary variable, equals 1 if the $b_{th}$ FS on $e^s$ is \\ &  assigned to $e^r$\\

         ${\delta_{e^s,m}^{e^r}}$   &Binary variable, equals 1 if any FS of $m_{e^s}$ on $e^s$ \\ &is assigned to $e^r$\\

         \bottomrule

    \end{tabular}
    \label{tab_notions}
\end{table}

\subsection{VNE Process}
In this paper, we only consider transparent VNE in which not only the number of FSs on each optical link is same but also the number of FSs requested by each virtual link from a VNR is same \cite{ChatterjeeSO15}.
A VNE example is shown in Fig. \ref{fig_VNE_model}, where three layers are considered.
The bottom layer is the physical infrastructure layer, from which physical resources are abstracted to form a pool, working at the middle layer.
The top layer contains VNRs being satisfied by allocating PNs and PLs.
A VNE process includes VNoE and VLiE.
In Fig. \ref{fig_VNE_model}, the VNR$_1$ has the VNoE result $\{a\rightarrow E, b\rightarrow B  \}$, and {VNR}$_2$ has $\{c\rightarrow C, d\rightarrow B, e\rightarrow D\}$.
Additionally, the VNR$_1$ has the VLiE result $\{(a,b)\rightarrow \{(E,A),(A,B)\} \}$ using FSs (\#1, \#2), and VNR$_2$ has $\{(d,c)\rightarrow \{(B,C)\}, (c,e)\rightarrow \{(C,D) \}\}$ using FSs (\#4, \#5, \#6).

\subsection{InP's Revenue and Cost}
The InP will obtain revenue from fulfilling each VNR, shown as
\begin{equation}\label{eq-InP-R}
    \begin{split}
    {\mathbb{R}({G^r})} = {\sum\limits_{n^r \in N^r}}{ \left[\alpha C(n^r) + \kappa W(n^r)\right]} +
                       {\sum\limits_{e^r \in E^r}}{\gamma B(e^r)},
    \end{split}
\end{equation}
where ${\alpha}$, ${\kappa}$ and ${\gamma}$ represent the prices to be charged for per unit of computing resource, radio resource and optical spectrum resource, respectively.

The embedding cost is concerned with the physical resources, which contain not only the resources required by ${G^r}$, but also some resource fragments.
In VNoE, the consumed physical resources include not only the computing and radio resources allocated to ${n^r}$ but also the wasted resource caused by the resulting imbalance in the remaining two types of resources.
As shown in Fig. \ref{fig_VNE_model}(b), it is difficult to utilize the remaining resources at nodes B and C due to the imbalanced resource utilization.
We define the level of imbalance (LoI) at $n^s$ as
\begin{equation}\label{eq-LoI}
    {\varrho(n^s)} = \left | \frac{C_l(n^s)}{C(n^s)} -  \frac{W_l(n^s)}{W(n^s)} \right |,
\end{equation}
where $C_l(n^s)$ and $W_l(n^s)$ represent the occupied computing resource and wireless channels of $n^s$, respectively.
After $n^r$ is embedded, the LoI of $n^s$ changes into
\begin{equation}\label{eq-LoI-change}
\begin{split}
    {\varrho'(n^s)} = \left | \frac{C_l(n^s) + C(n^r)}{C(n^s)} -
       \frac{ W_l(n^s) + W(n^r)}{W(n^s)} \right |.
\end{split}
\end{equation}

The increase in LoI can be calculated as
\begin{equation}\label{eq-LoI-increase}
    {\Delta \varrho(n^s)} =
        \begin{cases}
        \varrho'(n^s) - \varrho(n^s), & \text{if ${ \varrho'(n^s) - \varrho(n^s) > 0}$,} \\
        0, & \text{otherwise.}
        \end{cases}
\end{equation}

Therefore, at the PNs involved in this embedding, the cost of the total resources actually consumed by ${N^r}$ is
\begin{equation}\label{eq-InP-Cn}
    \begin{split}
    {\mathbb{C}_n(G^r, x_{n^s}^{n^r})} = &{\sum\limits_{n^r \in N^r}}{\sum\limits_{n^s \in N^s}}{x_{n^s}^{n^r} \cdot (1 + \Delta \varrho(n^s))} \\
     &\cdot \left[\alpha' C(n^r) + \kappa' W(n^r)\right] ,
    \end{split}
\end{equation}
where binary variable $x_{n^s}^{n^r}$ indicates whether ${n^r}$ is embedded onto ${n^s}$ (${x_{n^s}^{n^r} = 1}$) or not (${x_{n^s}^{n^r} = 0}$), ${\alpha'}$ and ${\kappa'}$ represent the cost of per unit of computing and radio resources, respectively.

\begin{figure}[ht]
    \centering
    \includegraphics[width=80mm]{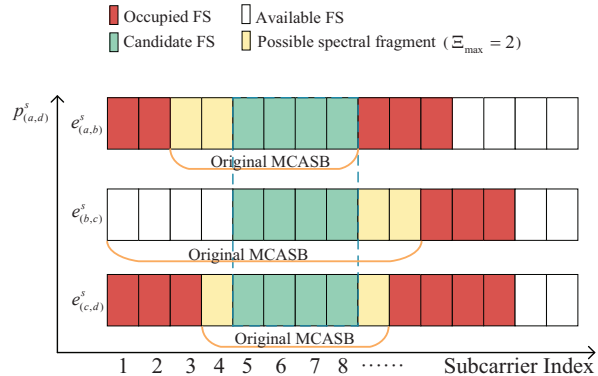}
    \caption {Examples of possible spectral fragments caused by embedding. }
    \vspace{0cm}
    \label{fig_Newly_FS}
\end{figure}

In VLiE, the consumed optical spectrum includes not only the FSs allocated to VLs but also some spectral fragments.
A spectral fragment $\xi_{e^s}^k$ is defined as an MACSB whose length is not more than the maximum fragment size ($f(m_{e^s})\leq \Xi_{max}$).
After a physical lightpath is established for $e^r$, spectral fragments may appear on the MACSBs where the FSs are assigned to $e^r$ along the physical lightpath.
Fig. \ref{fig_Newly_FS} shows examples of possible spectral fragments caused by embedding. It is seen that the spectral fragments appear either at one side of the allocated FSs or at both side of them. The number of FSs in the fragments introduced by the embedding of ${e^r}$ is denoted as $\xi_{e^s}^{e^r}$.

The cost of total spectrum resource consumed by ${E^r}$ is
\begin{small}
\begin{equation}\label{eq-InP-Ce}
    {\mathbb{C}_e(G^r, y_{p^s}^{e^r}, z_{e^s,b}^{e^r})} =  {\sum\limits_{e^r \in E^r}}{\sum\limits_{e^s \in E^s}}
    { y_{p^s}^{e^r}  I_{e^s}^{p^s}  \gamma'  \left [ \sum\limits_{b = 1}^{B(e^s)} z_{e^s,b}^{e^r} + \xi_{e^s}^{e^r}\right]},
\end{equation}
\end{small}
where binary variable $y_{p^s}^{e^r}$ indicates whether ${e^r}$ is embedded onto physical path $p^s$ (${y_{p^s}^{e^r} = 1}$) or not ($y_{p^s}^{e^r} = 0$), binary variable $z_{e^s,b}^{e^r}$ indicates whether the $b_{th}$ FS on optical link $e^s$ is assigned to $e^r$ ($z_{e^s,b}^{e^r} = 1$) or not ($z_{e^s,b}^{e^r} = 0$) and ${\gamma'}$ represents the cost of per unit of spectrum resource.

\subsection{VNE Problem Formulation}
The profit of all VNRs is affected by the number of VNRs accepted successfully, the revenue obtained from fulfilling VNRs and the embedding cost for VNoE and VLiE. Therefore, the embedding profit from all accepted VNRs is
\begin{small}
\begin{equation}\label{eq-InP-P}
 \begin{aligned}
        \mathbb{P}(&\Upsilon,\ x_{n^s}^{n^r},\  y_{p^s}^{e^r},\  z_{e^s,b}^{e^r}) \\
        = &{\sum\limits_{G^r \in \Upsilon}} {\nu^r} \left\{\mathbb{R}(G^r) - \mathbb{C}_n(G^r, x_{n^s}^{n^r}) - \mathbb{C}_e(G^r, y_{p^s}^{e^r}, z_{e^s,b}^{e^r}) \right\},
 \end{aligned}
\end{equation}
\end{small}
where binary variable $\nu^r$ indicates whether VNR ${G^r}$ is successfully embedded onto physical network ${G^s}$ (${\nu^r} = 1$) or not (${\nu^r} = 0$).

The following constraints apply to the VNE problem \cite{GuanZL0NR19}:

Each VN cannot be assigned to more than one PN. A PN cannot be assigned to two different VNs in the same VNR at the same time,
      \begin{small}
      \begin{equation}\label{eq-st-n-121}
        \text{C1}: \sum\limits_{{n^s} \in {N^s}} {x_{n^s}^{n^r}}  = 1, \quad \forall {n^r} \in {N^r},
      \end{equation}
      \begin{equation}\label{eq-st-n<1}
        \text{C2}: \sum\limits_{{n^r} \in {N^r}} {x_{n^s}^{n^r}} \le 1, \quad \forall {n^s} \in {N^s}.
      \end{equation}
      \end{small}

The PN assigned to the VN should have sufficient available resources,
\begin{small}
       \begin{equation}\label{eq-st-n-cc}
        \text{C3}: \sum\limits_{{n^s} \in {N^s}} {x_{n^s}^{n^r}}  \cdot C_a(n^s) \ge C(n^r), \quad  \forall {n^r} \in {N^r},
      \end{equation}
      \begin{equation}\label{eq-st-n-wc}
        \text{C4}: \sum\limits_{{n^s} \in {N^s}} {x_{n^s}^{n^r}}  \cdot W_a(n^s) \ge W(n^r), \quad  \forall {n^r} \in {N^r}.
      \end{equation}
\end{small}

A VN can only be assigned to a PN that is within its preferred aera,
      \begin{small}
       \begin{equation}\label{eq-st-n-lc}
        \begin{aligned}
            \text{C5}: \sum\limits_{{n^s} \in {N^s}} {x_{n^s}^{n^r}}  \Delta loc(n^r) \ge dis[loc(n^r),loc(n^s)],
            \forall {n^r} \in {N^r}.
        \end{aligned}
       \end{equation}
    \end{small}

 If $p^s$ is the light path established for a VL $e^r$, the two end nodes of $p^s$ must host the two end nodes of $e^r$.
    \begin{small}
      \begin{equation}\label{eq-st-l-endnodes}
        \text{C6}: y_{p^s}^{e^r} = x_{s(p^s)}^{s(e^r)} \cdot x_{t(p^s)}^{t(e^r)}, \quad \forall {e^r} \in {E^r}, \forall {p^s} \in {P^s}.
      \end{equation}
      \end{small}

The available spectrum resource of the physical lightpath allocated to a VL should not be less than the link's resource requirement,
     \begin{small}
      \begin{equation}\label{eq-st-l-bc}
        \begin{aligned}
            \text{C7}: \sum\limits_{{p^s} \in {P^s}} {y_{p^s}^{e^r} \cdot {B(p^s)}}  \ge {B(e^r)}, \quad \forall {e^r} \in {E^r}.
        \end{aligned}
      \end{equation}
    \end{small}

On each PL allocated to a VL, the number of the involved FSs should match the embedding requirement,
    \begin{small}
      \begin{equation}\label{eq-st-l-fsc}
        \begin{aligned}
            \text{C8}: \sum\limits_{b = 1}^{B(e^s)} {z_{e^s,b}^{e^r}}  = y_{p^s}^{e^r} \cdot I_{e^s}^{p^s} \cdot {B(e^r)}, \quad \forall {e^r} \in {E^r},
            \forall {e^s} \in {E^s}.
        \end{aligned}
       \end{equation}
    \end{small}

The FSs allocated to each VL must neighbor each other,
    \begin{small}
      \begin{equation}\label{eq-st-FS-neighbor}
            \begin{aligned}
            \text{C9}: t_b(e^r) - s_b(e^r) + 1 = B(e^r), \quad \forall {e^r} \in {E^r}.
            \end{aligned}
       \end{equation}
       \end{small}

The selected contiguous FSs on each optical link involved in the lightpath established for a VL should have the same indices,
     \begin{small}
      \begin{equation}\label{eq-st-FS-aligned}
        \begin{split}
            \text{C10}: \sum\limits_{b = 1}^{B(e^s)} {(z_{e^s,b}^{e^r} - z_{{e'}^s,b}^{e^r})} = 0, \text{ if }y_{p^s}^{e^r} \cdot I_{e^s}^{p^s} \cdot I_{{e'}^s}^{p^s} = 1,  \\ \quad \forall {e^r} \in {E^r}, \forall {e^s, {e'}^s} \in {E^s}.
        \end{split}
      \end{equation}
    \end{small}

The selected FSs by each VL must not conflict, that is, there is no overlapping FS for any two different VLs sharing common optical links,
   \begin{small}
    \begin{equation}\label{eq-st-FS-nonoverlapping}
        \begin{split}
            \text{C11}: \mathbb{I} \left(s_b(e^r) \leq s_b({e'}^r)\right) = \mathbb{I} \left( t_b(e^r) \geq t_b({e'}^r)\right ), \\
            \text{ if } y_{p^s}^{e^r} \cdot y_{p^s}^{{e'}^r} = 1, \forall {e^r, {e'}^r} \in {E^r}.
        \end{split}
    \end{equation}
    \end{small}

We jointly optimize both VNoE and VLiE to maximize the embedding profit from all the accepted VNRs. The joint VNoE and VLiE problem is expressed as follows:
\begin{equation}\label{eq-Problem-P}
 \begin{aligned}
    \mathcal{P: }\quad &\max \limits_{{ \textbf{x, y, z}}}\   {\sum\limits_{G^r \in \Upsilon}} {\nu^r} \left\{\mathbb{R}(G^r) - \mathbb{C}_n(G^r, \textbf{x}) - \mathbb{C}_e(G^r, \textbf{y}, \textbf{z}) \right\} \\
                 &\mathrm{ s.t. }  \ \text{C1} - \text{C11}.
 \end{aligned}
\end{equation}

We can observe that the problem $\mathcal{P}$ is a integer nonlinear optimization problem.
It is challenging to solve $\mathcal{P}$ if using conventional optimization techniques.
Furthermore, the performance of VNoE cannot be fully evaluated until the corresponding VLiE is complete.
Therefore, in order to optimize VNE, all the available VNoE options need to be considered and consequently all the corresponding VLiE options need to be listed for assessment. The most appropriate VNoE and VLiE decision will be made based on its profit contribution. In brief, VNoE and VLiE are closely coupled and mutually dependent.
Therefore, there are two issues in efficiently solving $\mathcal{P}$.
\begin{enumerate}
    \item Is it possible to convert this nonlinear integer optimization problem into another straightforward one?
    \item How can we take into account the interdependence of VNoE and VLiE?
\end{enumerate}


\section{Proposed Algorithm}
\label{sec:proposal}

Considering the two aforementioned difficulties, we convert $\mathcal{P}$ into a bilevel optimization problem. In this study, the VLiE is regarded as the lower-level optimization problem with the goal of maximizing profit by leasing optical spectrum, whereas the VNoE is regarded as the upper-level optimization problem with the goal of maximizing profit by leasing the total resources.
As a result, $\mathcal{P}$ can be changed into a bilevel problem as

\begin{equation}\label{eq-Problem-P1}
    \begin{aligned}
    \mathcal{P}1: \  &\max \limits_{{ \textbf{x}, \textbf{y}, \textbf{z}}}\   {\sum\limits_{G^r \in \Upsilon}} {\nu^r} \left\{\mathbb{R}(G^r) - \mathbb{C}_n(G^r, \textbf{x}) - \mathbb{C}_e(G^r, \textbf{y}, \textbf{z}) \right\} \\
                 &\mathrm{ s.t.}\  (\textbf{y}, \textbf{z}) \in \mathop{\arg\max}_{\textbf{y}, \textbf{z}} \Big\{- \sum\limits_{G^r \in \Upsilon} {\nu^r} \mathbb{C}_e(G^r, \textbf{y}, \textbf{z}) : \text{C6}-\text{C11} \Big\} \\
                 &\quad \ \  \text{C1} - \text{C5}.
    \end{aligned}
\end{equation}

In order to prove the relationship between $\mathcal{P}$ and $\mathcal{P}1$, Lemma \ref{lemma1} must first be introduced \cite{HuangWWL20}.

\begin{lemma}\label{lemma1}
The optimal solution to $\mathcal{P}$ is a feasible solution to $\mathcal{P}1$.
\end{lemma}

\begin{proof}\label{proof1}
Denote the optimal solution to the original problem $\mathcal{P}$ as $\{\textbf{x}^{*}, \textbf{y}^{*}, \textbf{z}^{*}\}$. Since $\{\textbf{x}^{*}, \textbf{y}^{*}, \textbf{z}^{*}\}$ satisfies C1 $-$ C11, the problem we need to prove becomes that $\{\textbf{y}^{*}, \textbf{z}^{*}\}$ is the optimal solution to the lower-level optimization problem of $\mathcal{P}1$.

Suppose the optimal solution to the lower-level optimization problem of $\mathcal{P}1$ is not $\{\textbf{y}^{*}, \textbf{z}^{*}\}$ but $\{\textbf{y}', \textbf{z}'\}$, then
\begin{equation}\label{eq-proof-c'>c*}
     - \sum\limits_{G^r \in \Upsilon} {\nu^r} \mathbb{C}_e(G^r, \textbf{y}', \textbf{z}') >  - \sum\limits_{G^r \in \Upsilon} {\nu^r} \mathbb{C}_e(G^r, \textbf{y}^{*}, \textbf{z}^{*}).
\end{equation}

Obviously, $\{\textbf{x}^{*}, \textbf{y}', \textbf{z}'\}$ is a feasible solution to $\mathcal{P}$. If we substitute (\ref {eq-proof-c'>c*}) into $\mathcal{P}$, then
\begin{equation}\label{eq-proof-p'>p*}
    \begin{aligned}
     &{\sum\limits_{G^r \in \Upsilon}} {\nu^r} \left\{\mathbb{R}(G^r) - \mathbb{C}_n(G^r, \textbf{x}^{*}) - \mathbb{C}_e(G^r, \textbf{y}', \textbf{z}') \right\}\\
     & > {\sum\limits_{G^r \in \Upsilon}} {\nu^r} \left\{\mathbb{R}(G^r) - \mathbb{C}_n(G^r, \textbf{x}^{*}) - \mathbb{C}_e(G^r, \textbf{y}^{*}, \textbf{z}^{*}) \right\}.
     \end{aligned}
\end{equation}

This clearly contradicts the assumption that $\{\textbf{x}^{*}, \textbf{y}^{*}, \textbf{z}^{*}\}$ is the optimal solution to $\mathcal{P}$. Thus, $\{\textbf{y}^{*}, \textbf{z}^{*}\}$ is proved to be the optimal solution to the lower-level optimization problem. This proves that the optimal solution to $\mathcal{P}$ is a feasible solution to $\mathcal{P}1$.
\end{proof}

\begin{theorem}\label{Theorem1}
$\mathcal{P}$ and $\mathcal{P}1$ have the same optimal solution.
\end{theorem}

\begin{proof}\label{proof2}
 Denote the optimal solution to $\mathcal{P}$ as $\{\textbf{x}^{*}, \textbf{y}^{*}, \textbf{z}^{*}\}$. By Lemma \ref{lemma1}, we can know that it is a feasible solution to $\mathcal{P}1$. Assume that $\{\textbf{x}', \textbf{y}', \textbf{z}'\}$ is the optimal solution to $\mathcal{P}1$, then
\begin{equation}\label{eq-proof-p*=p1*}
    \begin{aligned}
     &{\sum\limits_{G^r \in \Upsilon}} {\nu^r} \left\{\mathbb{R}(G^r) - \mathbb{C}_n(G^r, \textbf{x}') - \mathbb{C}_e(G^r, \textbf{y}', \textbf{z}') \right\}\\
     & > {\sum\limits_{G^r \in \Upsilon}} {\nu^r} \left\{\mathbb{R}(G^r) - \mathbb{C}_n(G^r, \textbf{x}^{*}) - \mathbb{C}_e(G^r, \textbf{y}^{*}, \textbf{z}^{*}) \right\}.
     \end{aligned}
\end{equation}

Obviously, this contradicts the assumption that $\{\textbf{x}^{*}, \textbf{y}^{*}, \textbf{z}^{*}\}$ is the optimal solution to $\mathcal{P}$. Thus, the optimal solution to $\mathcal{P}$ is the optimal solution to $\mathcal{P}1$. By the same token, it can be proved that the optimal solution to $\mathcal{P}1$ is also the optimal solution to $\mathcal{P}$.
\end{proof}

The modified problem $\mathcal{P}1$ provides some benefits in addition to maintaining the original problem $\mathcal{P}$'s optimal solution.
\begin{enumerate}
    \item The initial problem is convert into two more manageable combinatorial optimization problems with fewer decision variables.
    \item $\mathcal{P}1$ allows for comprehensive consideration of the interdependence between VNoE and VLiE. The quality of each VNoE result can be evaluated because the best VLiE result corresponding to each VNoE result has been produced in the lower level.

\end{enumerate}

\subsection{BiVNE}
To solve this transformed problem, we propose BiVNE, a nested bilevel VNE approach. \textbf{Algorithm \ref{Algorithm-General-BiVNE}} depicts the general framework of BiVNE. The VNRs are fulfilled one by one through BiVNE. During the initialization phase, the feasible candidate node set of each VN is pruned to a smaller scale.
The upper-level and lower-level optimizations are performed at each iteration of the main loop, with the latter nested within the former.
The ACS with a sorting approach is used in the upper-level optimization to generate a VNoE solution.
The Dijkstra algorithm and an exact-fit spectrum slot allocation method are used in the lower-level optimization to obtain the corresponding optimal embedding of VLs based on the provided VNoE result.
The performance of the whole VNE solution can be evaluated after the corresponding optimal VLiE of each VNoE is found~\cite{LiK14,LiFK11,LiKWTM13,CaoKWL12,CaoKWL14,LiDZZ17,LiKD15,LiDZK15,LiCFY19,LiKZD15,LiZKLW14,LiFKZ14}.
The above operation will be repeated until a predefined threshold is reached.
In the next sections, we will introduce BiVNE in detail.

\begin{algorithm}[ht]
 \caption{BiVNE}
  \label{Algorithm-General-BiVNE}
   \begin{algorithmic}[1]
        \STATE $gen = 0$;
        \STATE Obtain a feasible candidate PNs set using \textbf{Algorithm \ref{Algorithm-Can-NSet-Pruning}};  
        \WHILE {$gen < gen_{max}$}
            \STATE Establish a VNoE solution using \textbf{Algorithm \ref{Algorithm-Solution-Construction}}: $\textbf{x} = \{ x_{n^s}^{n^r} \ | \ n^r \in N^r, n^s \in N^s  \} $; 
            \STATE Apply \textbf{Algorithm \ref{Algorithm-link-embedding}} to calculate the optimal embedding of VLs and FS allocation based on the given $\textbf{x}$: $\textbf{y} = \{ y_{p^s}^{e^r}\ | \ e^r \in E^r, p^s \in P^s \}$, $\textbf{z} = \{ z_{e^s,b}^{e^r} \ | \ e^r \in E^r, e^s \in e^s, b < B(e^s)  \}$;
            \STATE Evaluate the total resource consumed by each VNoE solution with the corresponding optimal embedding of VLs and FS allocation;
            \STATE Apply the local search procedure to the iteration-best solution $(\textbf{x}^b, \textbf{y}^b, \textbf{z}^b)$;
            \STATE Based on (\ref{eq-phe-upd-glo}), update the global pheromone in ACS;
            \STATE $gen = gen + 1$;
        \ENDWHILE
        
        \STATE \textbf{Output:} the best VNoE solution with the corresponding optimal VLiE solution and FS allocation $(\textbf{x}^*, \textbf{y}^*, \textbf{z}^*)$.
  \end{algorithmic}
 \end{algorithm}

\subsection{Candidate Node Set Reduction}
In the upper-level optimization problem, the size of the solution space for each VN is determined by the total number of PNs.
For instance, if a VNR contains $|N^r|$ VNs, the solution space of this VNR is ${|N^s|}^{|N^r|}$.
Obviously, the upper-level solution space is so large that it has to be reduced.

As shown in \textbf{Algorithm \ref{Algorithm-Can-NSet-Pruning}}, a method to shrink the solution space is designed on the basis of three facts:
\begin{enumerate}[i)]
    \item \label {fact-1} When a VNR arrives, we can preselect a candidate node set $\Phi_N(n^r)$ for every $n^r$ according to its resource capacity requirement and preferred area: $\Phi_N(n^r) = \{n^s \in N^s \mid C_a(n^s) \ge C(n^r),\  W_a(n^s) \ge W(n^r), \ dis[loc(n^r),loc(n^s)] \le \Delta loc(n^r)\}$.
    \item \label {fact-2}For any PN $n^s$ in the candidate node set $\Phi_N(n^r)$ selected by \ref{fact-1}), if it has fewer attached links than $n^r$, it is not applicable for $n^r$.
    \item \label {fact-3}The PN is still not applicable if the amount of spectrum provided by its attached links cannot meet the requirement of the corresponding VLs.
\end{enumerate}

Thus, the PNs in cases \ref{fact-2}) and \ref{fact-3}) are removed from the original candidate set. In order to make the reduction, the number of the available attached links of $n^s$ is first calculated as
\begin{equation}\label{eq-dl(n^s)}
  d_l(n^s) = \sum\limits_{e_n^s \in E(n^s)} \phi_{n^s}^{e^s},
\end{equation}
where $E(n^s) = \{e^s \in E^s \mid s(e^s) = n^s \parallel t(e^s) = n^s\}$ is the set of attached links of PN $n^s$, $\phi_{n^s}^{e^s} = \mathbb{I} (f(m_{e_n^s}^{max}) \geq B(e^r))$ is a binary variable indicating whether attached link $e_n^s$ has enough available contiguous FSs to meet the bandwidth requirement or not, and $m_{e_n^s}^{max}$ is the MACSB with the largest length on $e^s$.
As a result, the candidate node set of each VN $n^r$ should be reduced by the following constraint
\begin{equation}\label{eq-ds>dr}
  d_l(n^s) \geq d(n^r),
\end{equation}
where $d(n^r)$ is the number of attached links of $n^r$.

\textbf{Algorithm \ref{Algorithm-Can-NSet-Pruning}} depicts the detailed procedure of set reduction. Based on (\ref{eq-st-n-121}) - (\ref{eq-st-n-lc}) and (\ref{eq-ds>dr}), each PN in the candidate node set is evaluated separately. If any PN in the candidate node set does not satisfy (\ref{eq-st-n-121}) - (\ref{eq-st-n-lc}) and (\ref{eq-ds>dr}), it is considered infeasible and removed from the alternative set. After the reduction, each candidate node set has a size smaller than or equal to $|N^s|$. In consequence, the size of the search space is significantly reduced.

\begin{algorithm}[ht]
 \caption{Candidate Node Set Reducing}
    \label{Algorithm-Can-NSet-Pruning}
     \begin{algorithmic}[1]
       \FOR {all $n^r$ in $N^r$}
            \STATE $\Phi_N(n^r) = N^s$,
            \FOR {all $n^s$ in $\Phi_N(n^r)$}
                \IF {$C_a(n^s) < C(n^r)$ or $W_a(n^s) < W(n^r)$ or $dis[loc(n^r),loc(n^s)] > \Delta loc(n^r)$ or $d_l(n^s) < d(n^r)$}
                 \STATE $\Phi_N(n^r) \leftarrow \Phi_N(n^r) \setminus \{n^s\}$,
                \ENDIF
            \ENDFOR
       \ENDFOR
   \RETURN $\Phi_N(n^r)$
  \end{algorithmic}
 \end{algorithm}

\subsection{Upper-Level Optimization}
Normally, the InP will earn greater profit by consuming less resources to fulfill the requirement of each VNR.
Therefore, the upper-level optimization goal can be simplified to minimize the sum of all resources allocated to any VNR and the resulting resource fragments.
The VNoE problem is a combinational optimization problem with the constraint that two VNs in a VNR cannot be embedded onto the same PN \cite{HernandoML16}.
To solve the upper-level problem, we propose an ACS-based algorithm. ACS is chosen for the following reasons:
\begin{enumerate}
    \item Evolutionary computation algorithms are widely applied to combinational optimization problems. Especially, ACS, as a variant of the ant colony optimization algorithm, has demonstrated higher strengths than other evolutionary algorithms in solving real-world problems.
    \item Unlike other evolutionary computation algorithms generating solutions simultaneously, the ACS optimizers construct solutions one by one. By this way, once a PN is assigned to a VN, it will not be assigned again to any other VN in the same VNR.
\end{enumerate}

The proposed ACS-based algorithm BiVNE consists of four operations: solution construction, fitness evaluation, local search and pheromone management.

\subsubsection{Solution Construction} \label{ssubsec-solu-constr}

\begin{algorithm}[ht]
 \caption{VNoE Solution Construction}
  \label{Algorithm-Solution-Construction}
   \begin{algorithmic}[1]
   \STATE Generate the embedding order of VNs based on the sorting strategy;
   \STATE Sort the VNs, and get the sorted set ${N^r}'$;
   \STATE $\textbf{x} = \emptyset$; \label{alg-in-SloCons-X=0}
   \FOR {all $n^r$ in ${N^r}'$}
       \IF {$\Phi_N(n^r) \neq \emptyset$ } 
            \STATE Calculate the probability that each candidate PN in $\Phi_N(n^r)$ based on (\ref{eq-probability});
            \STATE Select a PN $n^s$ from $\Phi_N(n^r)$ based on (\ref{eq-state-transition-rule}) and set $x_{n^s}^{n^r} = 1$;
            \STATE Delete $n^s$ from candidate node set of remaining VNs in the same VNR;
            \STATE Update the local pheromone in ACS based on (\ref{eq-phe-upd-loc});
       \ENDIF
   \ENDFOR
   \RETURN $\textbf{x}$
  \end{algorithmic}
 \end{algorithm}

During each iteration, ants search for all the possible solutions and then choose the optimal one, i.e. the VNoE solution.
By allocating the PNs in the candidate node set to the VNs one by one at a time, each ant builds a potential VNoE solution.
The request of embedding a VNR is rejected if any involved VN's corresponding candidate node set is empty.
The complete procedure of VNoE solution construction is shown in \textbf{Algorithm \ref{Algorithm-Solution-Construction}}.

\begin{figure}[ht]
    \centering
    \subfloat[Random.]{\label{fig_virtual_node_order-random}
    \includegraphics[width=3.5cm]{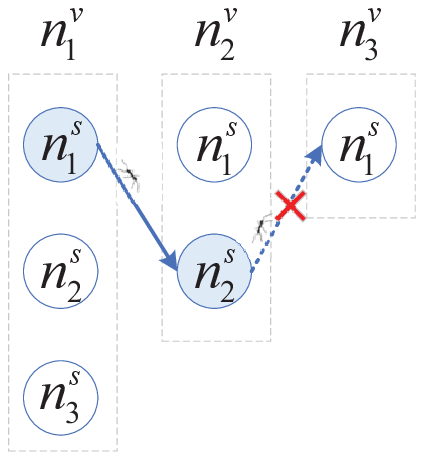}
    }
    \subfloat[Sorted.]{\label{fig_virtual_node_order-sorted}
    \includegraphics[width=3.5cm]{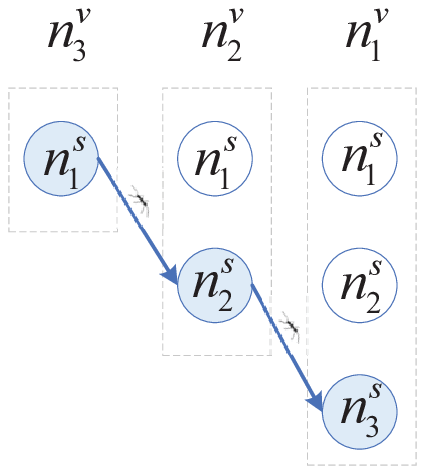}
    }
    \caption{Illustration of operating principle of proposed sorting strategy.}
    \label{fig_virtual_node_order-subs}       
\end{figure}

The VNoE is generally performed by ants in a random order \cite{LiuZDLGZ18, LiangZZZ20}.
However, it is not the best way in solving this problem.
Fig. \ref{fig_virtual_node_order-random} shows an example of performing the embedding randomly.
We assume three VNs (i.e., $n_1^v$, $n_2^v$, $n_3^v$) need to be embedded.
The numbers of PNs in the candidate node sets of nodes $n_1^v$, $n_2^v$, and $n_3^v$ are 3, 2, and 1, respectively.
If the three VNs are embedded in this order, the solution will be infeasible in most cases due to the difficulty of embedding $n_3^v$.
As shown in Fig. \ref{fig_virtual_node_order-random}, the ant assigns a PN for $n_1^r$ first, and the corresponding result is $\{n_1^r \rightarrow n_1^s\}$.
Then, the ant only can assign $n_2^s$ to $n_2^v$ due to constraint C2.
Thus, there is no PN for $n_3^v$ to select.
We can easily find that $n_3^v$ cannot be embedded successfully as long as $\{n_1^r \rightarrow n_1^s\}$ or $\{n_2^r \rightarrow n_1^s\}$.
 The reason is that the VNs with fewer candidate nodes are embedded later.
So, to construct a feasible solution efficiently, we propose a sorting strategy by reranking the VNs in ascending order of the number of their candidate nodes.
As a result, ${N^r}$ becomes a new set ${N^r}'$.
Fig. \ref{fig_virtual_node_order-sorted} depicts an implementation of performing the proposed sorting strategy, where ${N^r}' = \{n_3^v, n_2^v, n_1^v \}$.
In this way, the feasible embedding solution for ${N^r}'$ can easily be found, and the result is $\{ n_3^v \rightarrow n_1^s, n_2^v \rightarrow n_2^s, n_1^v \rightarrow n_3^s \}$.

 As shown in \textbf{Algorithm \ref{Algorithm-Solution-Construction}}, after performing the sorting strategy, the VNs are embedded in turn.
First, the matrix of VNoE result is initialized as $\textbf{x} = \emptyset$ (Line \ref{alg-in-SloCons-X=0}).
Then, for $n^r \in {N^r}'$, if its candidate PN set $\Phi_N(n^r)$ is not empty, the probability that each candidate PN in $\Phi_N(n^r)$ to be selected is given as \cite{DorigoG97}
 \begin{equation}\label{eq-probability}
 \begin{aligned}
     pr(n^r, n^s) = \frac{\tau(n^r, n^s) \cdot {\eta(n^r, n^s)}^\beta}{\sum\limits_{{n^s}' \in \Phi_N(n^r)} \tau(n^r,{n^s}') \cdot {\eta(n^r, {n^s}')}^\beta},\\
    \forall n^r \in {N^r}', \forall n^s \in {\Phi_N(n^r)},
 \end{aligned}
 \end{equation}
where $\tau(n^r, n^s)$ is the pheromone, $\eta(n^r, n^s)$ is the heuristic information, and $\beta$ is a factor that determines the relative weight of $\tau(n^r, n^s)$ and $\eta(n^r, n^s)$.
The heuristic information is used as a subjective incentive for ants to search for a potentially good solution.
In BiVNE, we define it as
\begin{equation}\label{eq-heuristic-information}
   \eta(n^r, n^s) = \frac{1}{\mathbb{C}^I(n^r\rightarrow n^s)},
\end{equation}
where $\mathbb{C}^I(n^r \rightarrow n^s)$ is the increment of the cost of fulfilling the resource requirements of $n^r$ and its attached links with the already embedded VNs.
Thus, the smaller the cost increment, the higher the probability of assigning PN $n^s$ to $n^r$ in (\ref{eq-probability}). $\mathbb{C}^I(n^r \rightarrow n^s)$ is calculated as
\begin{equation}\label{eq-resource-increment-embed-VNN}
  \mathbb{C}^I(n^r \rightarrow n^s)= \mathbb{C}_n(x_{n^s}^{n^r}) + \mathbb{C}_e^{min}( x_{n^s}^{n^r}),
\end{equation}
where $\mathbb{C}_e^{min}( x_{n^s}^{n^r})$ is the minimum increment of cost for link resource consumption if $x_{n^s}^{n^r} = 1$, which can be calculated as
\begin{equation}\label{eq-Ce_min}
  \mathbb{C}_e^{min}( x_{n^s}^{n^r}) = \sum\limits_{n_i^r \in {N_{as}^r}'} h(n^s, \mathcal{M}(n_i^r)) \cdot B(G^r)  \cdot \gamma',
\end{equation}
where ${N_{as}^r}'$ is the set of VNs whose embedding orders are in front of $n^r$, $h(n^s, \mathcal{M}(n_i^r))$ is the length of the shortest path between $n^s$ and the PN $\mathcal{M}(n_i^r)$ host $n_i^r$, and $B(G^r)$ is the number of FSs required by every VL in corresponding VNR $G^r$.

Then, for $n^r$, an ant chooses a PN from the candidate node set ${\Phi_N(n^r)}$ according to the state transition rule based on the pseudorandom proportional given by
{\small
\begin{equation}\label{eq-state-transition-rule}
      n^s =
        \begin{cases}
            \mathop{\arg\max}\limits_{ {n^s}' \in \Phi_N(n^r)} \tau(n^r,{n^s}') \cdot {\eta(n^r, {n^s}')}^\beta, &\text{if ${q \leq q_0}$,} \\
            RWS(pr(n^r, n^s)),                            & \text{otherwise,} 
        \end{cases}
\end{equation}}
where $q$ is a uniformly distributed random number between $[0,1]$, $RWS(pr(n^r, n^s))$ is a PN chosen from $\Phi_N(n^r)$ by using the roulette wheel selection function $RWS(\cdot)$ on the basis of the probability distribution given in (\ref{eq-probability}), and $q_0$ is a parameter that controls the behaviors of ants in exploitation and exploration.

\subsubsection{Fitness Evaluation}
Following the construction of the VNoE solution, the corresponding optimal VLiE solution and FS allocation are obtained in the lower-level optimization, which will be discussed in Section \ref{subsection-lower}.
The performance of a completed solution comprised of the VNoE and VLiE results can then be evaluated~\cite{WuLKZZ19,WuLKZZ17,LiC23,WilliamsLM23,LyuYWHL23}.
The fitness function is the cost for fulfilling the total resource requirements of both VNoE and VLiE, and it is calculated as follows
\begin{equation}\label{eq-fitness}
   \mathbb{C}(\textbf{x},\  \textbf{y}^{*}, \textbf{z}^{*}) = \mathbb{C}_n(\textbf{x}) + \mathbb{C}_e(\textbf{y}^{*}, \textbf{z}^{*}),
\end{equation}
where $\mathbb{C}_n(\textbf{x})$ is the cost for fulfilling the resource requirements of VNs with the VNoE result matrix $\textbf{x}$, and $\mathbb{C}_e(\textbf{y}^{*}, \textbf{z}^{*})$ is the minimum cost corresponding to the optimal VLiE solution and FS allocation result $(\textbf{y}^{*}, \textbf{z}^{*})$ when $\textbf{x}$ is given.

\subsubsection{Local Search}
Following the evaluation of fitness, if the iteration-best solution $ \{\textbf{x}^{b}, \textbf{y}^{b}, \textbf{z}^{b}\}$ is feasible, a local search operation is carried out on it to speed up the convergence.
First, the order of VNs is generated as introduced in subsection \ref{ssubsec-solu-constr}.
Then, each VN is checked according to its order.
Besides the selected PN in the solution, any other PN belonging to the VN's candidate set is judged on whether it can contribute a lower fitness value~\cite{LiKWCR12,LiWKC13,CaoKWLLK15,LiDY18,WuKZLWL15,LiKCLZS12,LiDAY17,LiDZ15,LiXT19,GaoNL19,LiuLC19,LiZ19,KumarBCLB18,CaoWKL11,LiX0WT20,LiuLC20,LiXCT20,WangYLK21,ShanL21,LaiL021,LiLLM21,WuKJLZ17,LiCSY19,LiLDMY20,WuLKZ20,PruvostDLL020,XuLA22,LiLL22,ZhouLM22,ChenLTL22,Williams0M22,FanLT20}.
The original selected PN is replaced by the new one if it can contribute a lower fitness value.
Note that, to ensure constraint C2, the PN for substitution should not be assigned to the other VNs.
After that, the solution $ \{\textbf{x}^{b}, \textbf{y}^{b}, \textbf{z}^{b}\}$ is updated accordingly.

\subsubsection{Pheromone Management}
During the solution construction procedure, each time an ant finds a feasible assignment $n^s$ for $n^r$, the local pheromone is updated on pair $(n^r, n^s)$ accordingly to reduce the probability of other ants making the same assignment. The local pheromone updating rule is
\begin{equation}\label{eq-phe-upd-loc}
  \tau(n^r, n^s) = (1 - \varphi) \cdot \tau(n^r, n^s) + \varphi \cdot \tau_{0},
\end{equation}
where $\varphi$ is the pheromone decay coefficient and $\tau_{0}$ is the initial value of the pheromone.

At the end of each iteration, only the global best ant is allowed to release pheromone.
After all the ants complete their assignments, the best assignment can be found and labeled as $(\textbf{x}^b, \textbf{y}^b, \textbf{z}^b)$.
To increase the pheromone on the good potential assignments, the global pheromone updating rule is applied to $(\textbf{x}^b, \textbf{y}^b, \textbf{z}^b)$.
The rule of global pheromone updating is
\begin{small}
\begin{equation}\label{eq-phe-upd-glo}
  \tau(n^r, n^s) =
  \begin{cases}
        (1 - \rho) \cdot \tau(n^r, n^s) + \rho \cdot \Delta\tau, &\text{if ${x_{n^s}^{n^r} = 1}$,} \\
        \tau(n^r, n^s),                            & \text{otherwise,}
  \end{cases}
\end{equation}
\end{small}
where $\rho$ denotes the pheromone decay parameter, and
\begin{equation}\label{eq-Delta-tau}
  \Delta\tau = \frac{1}{\mathbb{C}(\textbf{x}^b, \textbf{y}^b, \textbf{z}^b)}.
\end{equation}
\subsection{Lower-Level Optimization}\label{subsection-lower}

Similar to the upper-level optimization, the aim of the lower-level optimization can be simplified to minimize the cost of fulfilling the resource requirement of VLs through optimizing the VLiE solution $\textbf{y}$ and the FS allocation $\textbf{z}$ under the given VNoE result $\textbf{x}$, which can be formulated as
\begin{equation}\label{eq-lower-objective}
 \begin{split}
    &\min \limits_{{ \textbf{y, z}}}
   {\sum\limits_{e^r \in E^r}}{\sum\limits_{e^s \in E^s}}
    { y_{p^s}^{e^r} \cdot I_{e^s}^{p^s} \cdot \gamma' \cdot \left [ \sum\limits_{b = 1}^{B(e^s)} z_{e^s,b}^{e^r} + \xi_{e^s}^{e^r}\right]}\\
    & \mathrm{s.t. }  \ \text{C6} - \text{C11}.
    \end{split}
\end{equation}

It can be seen from (\ref{eq-lower-objective}) that the longer the physical path in which a VL is embedded, and/or the more new spectrum fragments generated by this path, the greater the cost of InP.
On this basis, we propose \textbf{Algorithm \ref{Algorithm-link-embedding}} to embed VLs.

\begin{algorithm}[!t]
 \caption{VLiE and FS Assignment}
  \label{Algorithm-link-embedding}
   \begin{algorithmic}[1]
   \STATE ${G^s}' = G^s$; \label{alg-in-LiEmbed-Des-i}
    \FOR {all ${e^s}'$ in ${E^s}'$} \label{alg-in-LiEmbed-Des-s}
        \IF {$f(m_{{e^s}'}^{max}) < B(e^r)$}
             \STATE Delete ${e^s}'$ from ${G^s}'$;
        \ENDIF
    \ENDFOR  \label{alg-in-LiEmbed-Des-t}
    \FOR {all ${e^r}$ in ${E^r}$}
        \STATE Find the shortest path $p^s$ using Dijksra algorithm;
        \STATE Delete $p^s$ from ${G^s}'$;
    \ENDFOR
    \STATE $P^s_e = \{ p^s \}$;
    \IF {$f(m_{P^s_e}^{max}) < B(e^r)$}
        \STATE The VNR is rejected;
    \ELSE
        \STATE set $y_{p^s}^{e^r} = 1$;
        \STATE Find the optimal available continuous FSs $[s_b(e^r), t_b(e^r)]$ with the smallest size of newly-generated spectrum fragment, and set $z_{e^s, b}^{e^r} = 1,\  s_b(e^r) \leq b \leq t_b(e^r)$.
    \ENDIF

   \RETURN $\{\textbf{y}, \textbf{z}\}$
  \end{algorithmic}
 \end{algorithm}

As mentioned before, the VNE is transparent, in other words, all the VLs require the same amount of FSs in a VNR.
To improve search efficiency, in the beginning, we check every PL and remove those that cannot meet the spectrum requirements of the VLs (Line \ref{alg-in-LiEmbed-Des-s} - \ref{alg-in-LiEmbed-Des-t}).
The removing operation is performed on a replica of $G^s$, i.e. ${G^s}'$ (Line \ref{alg-in-LiEmbed-Des-i}), so the links in $G^s$ are not deleted, and the embedding result on ${G^s}'$ can be directly applied to $G^s$.
Afterwards, for each VL, a greedy algorithm is used to find the shortest path for it based on the Dijkstra algorithm.
If the size of MACSB on the path is smaller than the bandwidth requirement of the VL, the corresponding VNR is rejected~\cite{LiWKC13,RuanLDL20,SunL20,LiNGY22}.
Otherwise, the PLs on the shortest path are allocated to the VL.
Subsequently, an FS assignment method is used to select the optimal available continuous FS block which has the smallest size of newly-generated spectrum fragment.
And then, the corresponding FSs are allocated to the VL.
After the shortest path and optimal FSs are assigned to the VL, the VLiE is completed, and then, the shortest path is deleted from ${G^s}'$ to prevent spectrum overlapping (constraint C11).
 After the path selection and the FS assignment procedures, the result $\{\textbf{y}, \textbf{z}\}$ is generated.


\section{Experimental Setup}
\label{sec:settings}
In this section, evaluations are performed to assess the performance of the proposed BiVNE algorithm. We tested two sizes of physical networks:
a realistic Deutsche Telecom (DT) network with 14 nodes and 23 links, and a random network with 50 nodes and 166 links. All the PNs are arranged in a 1000*1000 coverage square.
The number of VNs contained in each VNR is randomly generated within a predefined range, see Table \ref{Simulation_parameters} for specific settings.
The probability of generating a VL between every two VNs is 0.5.
${\alpha = \kappa = \gamma =3}$, and ${\alpha' = \kappa' = \gamma' = 1}$ .

In the BiVNE algorithm, the population size is $NP = 10$, maximum generation number is $gen_{max} = 150$, and the other parameters are set in accordance with \cite{DorigoG97}:
$\beta = 2$, $q_0 = 0.9$, $\varphi = \rho = 0.1$, and $\tau_{0} = (|N^s| \cdot Cost)^{-1} $, where $Cost$ is the total cost produced by the greedy algorithm \cite{ZhaoSB13}.
Table \ref{Simulation_parameters} summarizes the remaining simulation parameters.

\begin{table}[t] \small   
\caption{Simulation Parameters}  
\centering  
   \begin{tabular}{c|lp{3cm}p{2cm}}   

        \toprule

        Parameters          &DT Topology    &Random Topology\\
        \midrule

         $|N^s|$            &14             &50\\
         $|E^s|$            &23             &166 \\
         $C(n^s)$           &[50,100]units  &[50,100]units \\
         $W(n^s)$           &[50,100]       &[50,100] \\
         $B(e^s)$           &[50,100]FS'    &[50,100]FS' \\
         $|N^r|$            &[3,4]          &[3,10] \\
         $\Delta loc(n^r)$  &[200,300]      &[200,300] \\
         $C(n^r)$           &[1,10]units    &[1,20]units \\
         $W(n^r)$           &[1,10]         &[1,20] \\
         $B(e^r)$           &[1,10]FS'      &[1,20]FS' \\

         \bottomrule

    \end{tabular}
    \label{Simulation_parameters}
\end{table}

To assess the performance of the designed algorithm, BiVNE, we compare it with three existing VNE algorithms: Greedy-SP-FF \cite{ZhaoSB13}, LRC-SP-FF \cite {GongWZL14}, and PL-KSP-FF \cite{FanXCCY21}.


\section{Empirical Studies}
\label{sec:experiments}

\begin{figure}[ht]
\centering
\subfloat[14-node DT topology.]{\label{fig-AR-DT} \includegraphics[width=6cm]{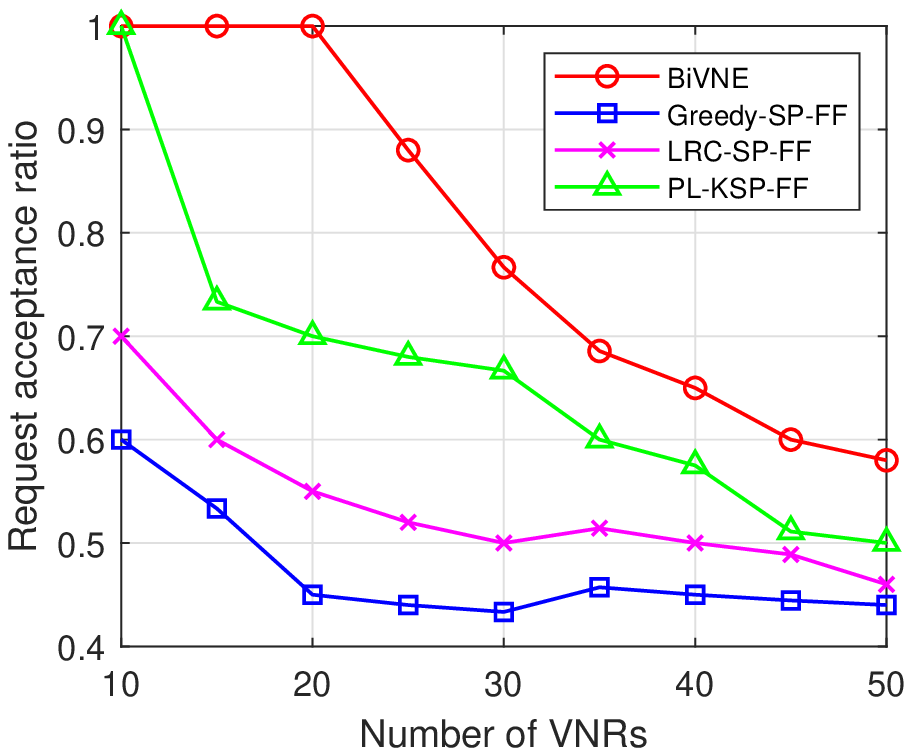}} 
\subfloat[50-node random topology.]{\label{fig-AR-50N} \includegraphics[width=6cm]{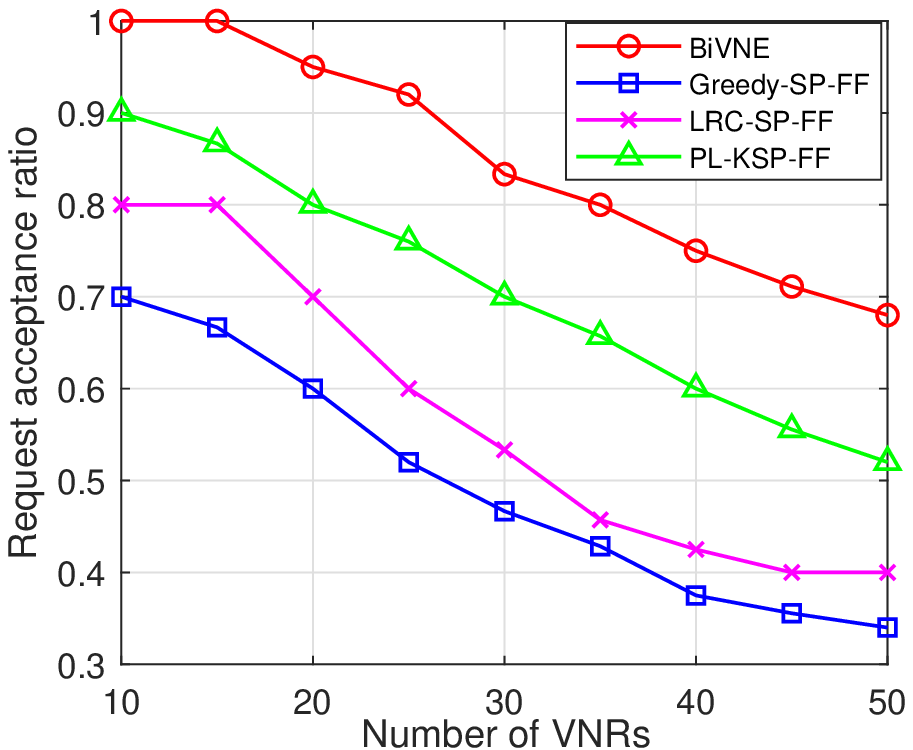}}
\caption{Performance comparison on acceptance ratio.}
\label{fig_AR}
\end{figure}

Fig. \ref{fig-AR-DT} and Fig. \ref{fig-AR-50N} compare the acceptance ratio of requests when the Greedy-SP-FF, the LRC-SP-FF, the PL-KSP-FF and the proposed BiVNE algorithms are implemented in the 14-node DT and the 50-node random networks, respectively.
The acceptance ratio in all the algorithms decrease with the increasing number of the VNRs in these two physical networks.
This is due to the remaining resources in the physical network are gradually exhausted as the number of accepted VNRs increases.
Hence, more VNRs may be rejected because of the lack of physical resources.
Moreover, when compared to the other three algorithms, the proposed BiVNE algorithm consistently provides the highest acceptance ratio.
The reason is that BiVNE improves the probability of generating feasible embedding solutions through not only reducing the candidate node set before embedding the VNRs, but also employing a sorting method for solution construction in the upper-level optimization.
Furthermore, the BiVNE algorithm takes into account the variation of LoI at the PN and the optical spectrum fragmentation on the PL, so that the VNRs are always satisfied at the cost of the smallest amount of resource fragments.
As a result, more VNRs can be accepted due to the fact that more available resources in the physical network can be provided.

\begin{figure}[ht]
\centering
\subfloat[14-node DT topology.]{\label{fig-ave-path-hops-DT}
\includegraphics[width=6cm]{{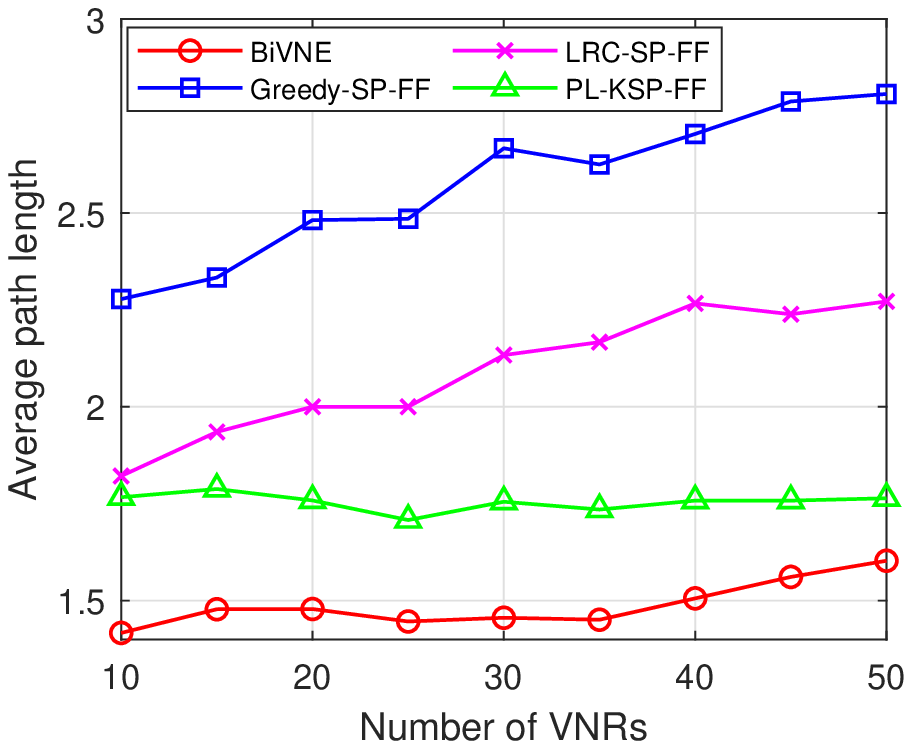}}}
\subfloat[50-node random topology.]{\label{fig-ave-path-hops-50N}
\includegraphics[width=6cm]{{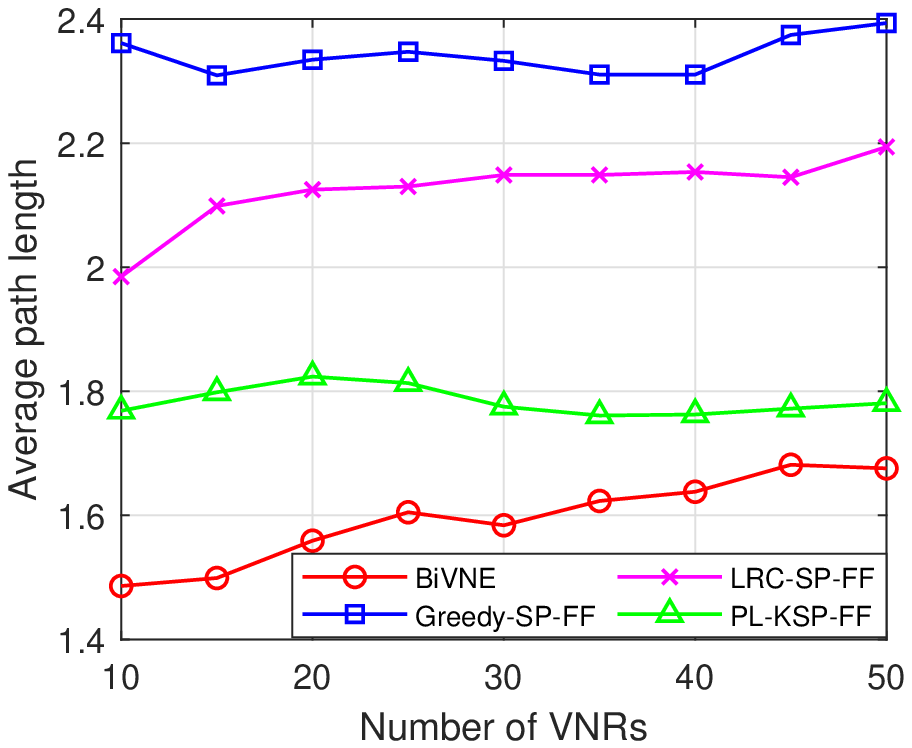}}}
\caption{Performance comparison on average path length.}
\label{fig_average_path_hops}
\end{figure}

Fig. \ref{fig-ave-path-hops-DT} and Fig. \ref{fig-ave-path-hops-50N} depict the performance comparison on the average path length obtained by conducting the four algorithms in the 14-node DT and the 50-node random networks, respectively.
We observe that the proposed BiVNE algorithm always obtains the shortest average path in any case in both the 14-node DT network (shown in Fig. \ref{fig-ave-path-hops-DT}) and the 50-node random network (shown in Fig. \ref{fig-ave-path-hops-50N}), thanks to the joint processing of VNoE and VLiE.
In contrast, the other three algorithms perform VNoE and VLiE separately.
PL-KSP-FF outperforms Greedy-SP-FF and LRC-SP-FF because, besides comparing the amount of resources provided by different PN candidates and their connected links, PL-KSP-FF takes the corresponding path length (i.e., the number of possible involved links) as an additional decision parameter in the procedure of VNoE.
However, the performance of PL-KSP-FF is still worse than that of the proposed BiVNE algorithm.
It is because, VLiE is performed under the only one resulting VNoE condition in the PL-KSP-FF algorithm, while more available VNoE choices are explored and provided for consideration when VLiE is decided in the BiVNE algorithm.
Apart from that, we observe that the average path length increases with the increasing number of VNRs in the proposed BiVNE algorithm, while this trend is not evident in the PL-KSP-FF algorithm.
The reason is, in addition to finding the shortest path like the PL-KSP-FF algorithm, the BiVNE algorithm tries to remain more continuous FSs in the physical networks during the embedding.
This leads to the fact that although the average path length in PL-KSP-FF is more stable, the proposed BiVNE algorithm is able to obtain a shorter average path while bringing a higher acceptance ratio than PL-KSP-FF.

\begin{figure}[ht]
\centering
\subfloat[14-node DT topology.]{\label{fig-R2C-DT}
\includegraphics[width=6cm]{{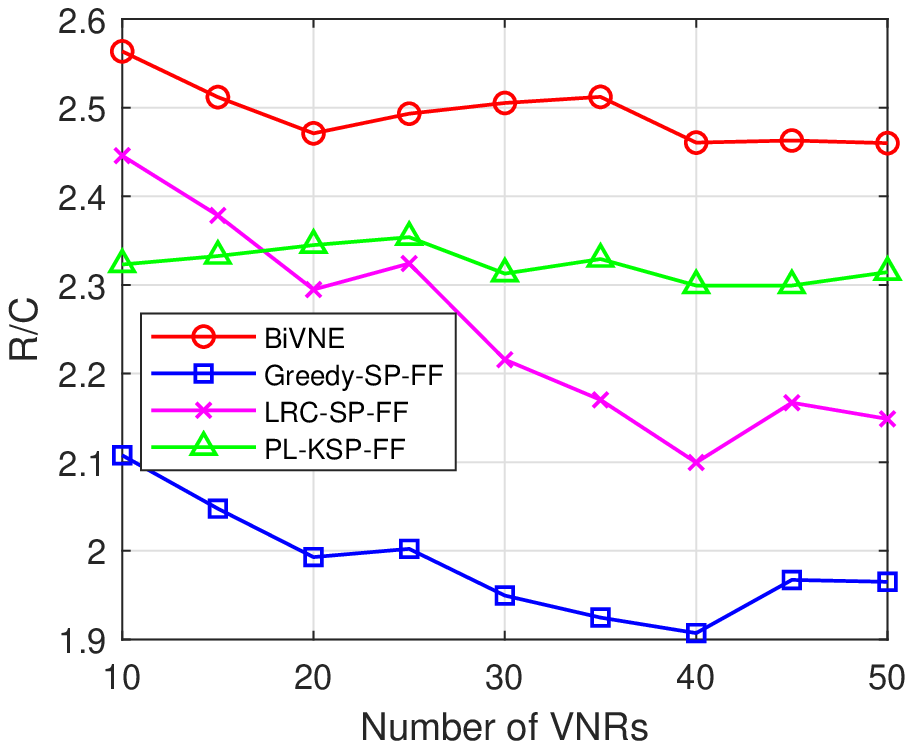}}}
\subfloat[50-node random topology.]{\label{fig-R2C-50N}
\includegraphics[width=6cm]{{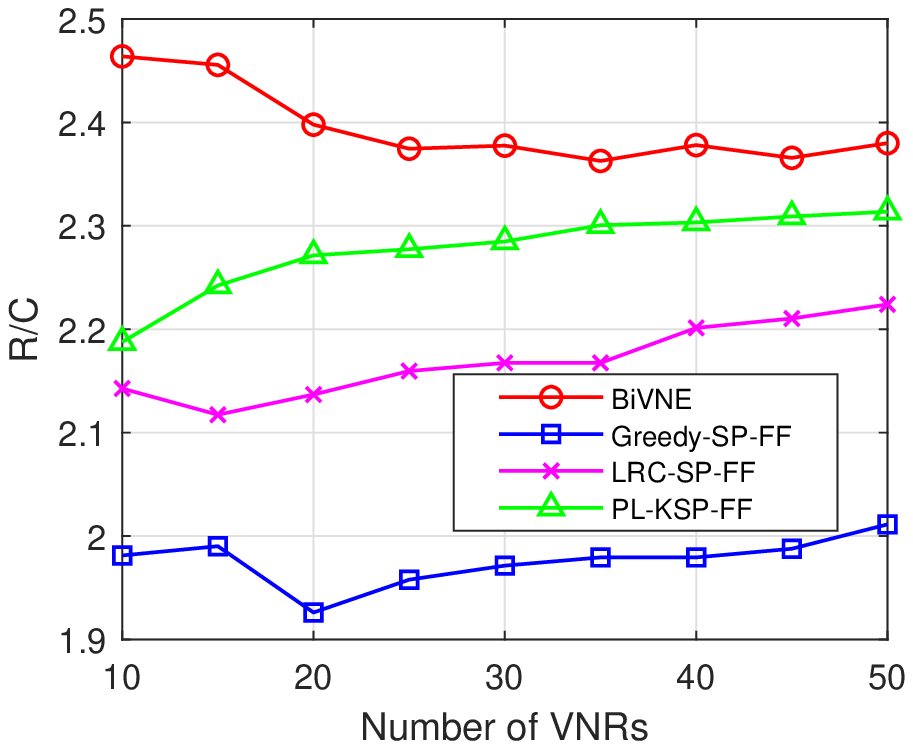}}}
\caption{Performance comparison on R/C ratio.}
\label{fig_R2C}
\end{figure}

Fig. \ref{fig-R2C-DT} and Fig. \ref{fig-R2C-50N} present the R/C ratio results when the four algorithms are implemented in the 14-node DT network and the 50-node random network, respectively.
When the number of VNRs is less than 20, the R/C ratio of PL-KSP-FF is worse than that of LRC-SP-FF, while the opposite is true when the number of VNRs is greater than 20.
The reason is that instead of striving for the minimum resource consumption of a single VNR, PL-KSP-FF tends to accept more VNRs by balancing the computing resource against bandwidth.
It can be observed that BiVNE always keeps the highest R/C ratio, which means that BiVNE can consume fewer resources to obtain the same revenue. The reason is that the BiVNE algorithm always aims to not only achieve the lowest LoI at the physical nodes, but also allocate the shortest path for the virtual link and allocate spectrum in a way that produces the smallest number of optical spectrum fragments.

\begin{figure}[ht]
\centering
\subfloat[14-node DT topology.]{\label{fig-total-profit-DT}
\includegraphics[width=6cm]{{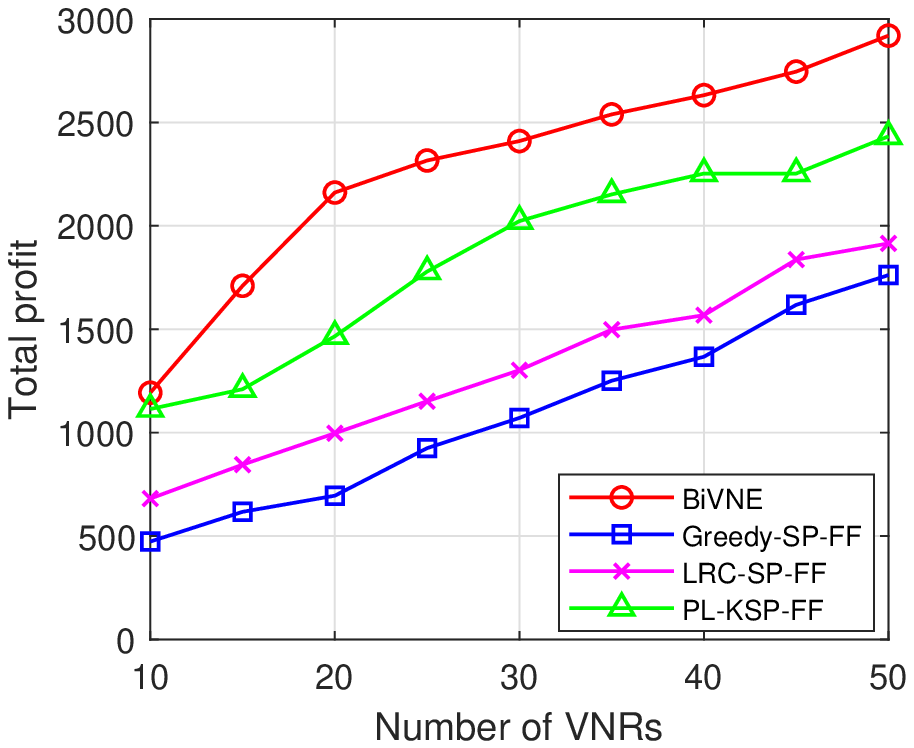}}}
\subfloat[50-node random topology.]{\label{fig-total-profit-50N}
\includegraphics[width=6cm]{{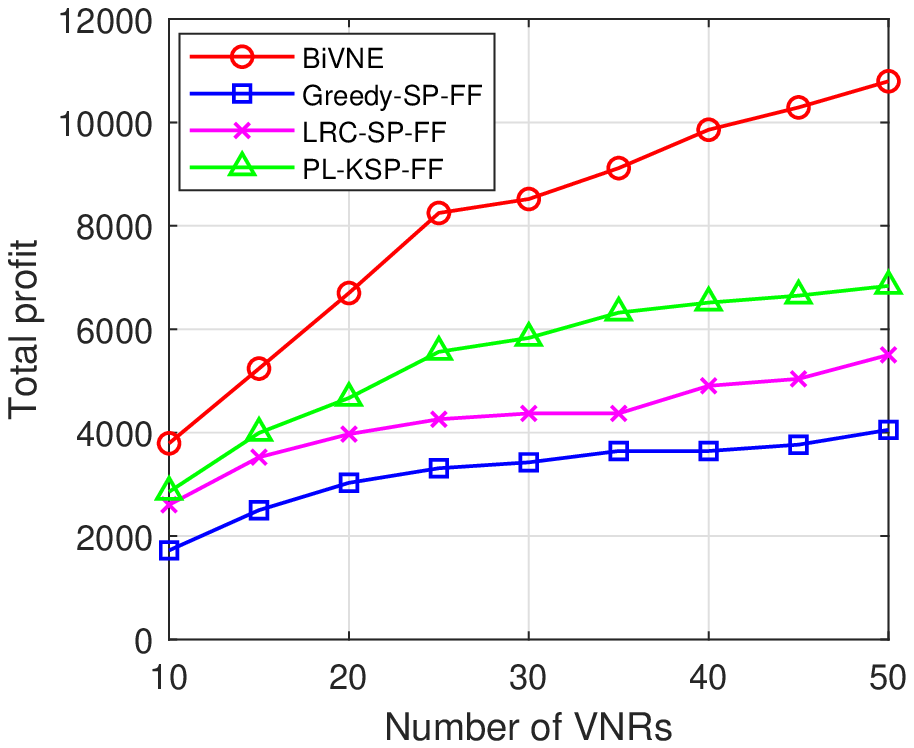}}}
\caption{Performance comparison on total profit.}
\label{fig_total_profit}
\end{figure}
Fig. \ref{fig-total-profit-DT} and Fig. \ref{fig-total-profit-50N} depict the total profit of the InP when the four algorithms are implemented in the 14-node DT network and the 50-node random network, respectively.
We observe that the total profit grows as satisfying more VNRs.
Compared with the three benchmark algorithms, a higher profit can always be achieved by the proposed BiVNE algorithm, benefitting from higher acceptance ratio and R/C ratio.
In addition, the superiority of BiVNE is more obvious in the 50-node random topology (Fig. \ref{fig-total-profit-50N}), especially when more VNRs arrive.
It illustrates our proposed algorithm plays better in large-scale networks.


\section{Conclusions and Future Directions}
\label{sec:conclusions}

In this paper, we studied the problem of multidimensional resource fragmentation-aware VNE in edge networks.
To quantify the resource fragments at PNs and PLs, we defined the LoI and a threshold for judging fragmentation.
To achieve efficient resource utilization, a coupled VNoE and VLiE problem was formulated.
The problem was transformed into a bilevel optimization problem with VNoE for the upper-level optimization problem and VLiE for the lower-level optimization problem.
To solve this problem, we proposed a nested method named BiVNE.
In the upper level of BiVNE, ACS was used to find the most promising VNoE result.
In the lower level of BiVNE, the optimal link resource allocation corresponding to each VNoE decision was determined using the Dijkstra algorithm and the exact-fit spectrum slot allocation approach.
Simulation results reveal that our proposed BiVNE method outperforms the existing algorithms in terms of acceptance ratio, R/C ratio and InP's profit.

\section*{Acknowledgment}

\bibliographystyle{IEEEtran}
\bibliography{IEEEabrv,your_bib}

\end{document}